\documentclass[12pt,a4paper,reqno,oneside]{amsart}
\usepackage{t1enc}
\usepackage{amssymb}
\usepackage[mathscr]{euscript}
\usepackage{color}
\usepackage{float}
\usepackage{graphicx}
\usepackage{url}
\usepackage{caption}
\usepackage{subcaption}

\numberwithin{equation}{section}

\newtheorem{defi}{Definition}[section]
\newtheorem{thm}[defi]{Theorem}

\newtheorem{rem}[defi]{Remark}

\newtheorem{cor}[defi]{Corollary}

\newcommand{\PP}{\mathbb{P}}
\newcommand{\QQ}{\mathbb{Q}}
\newcommand{\E}{\mathbb{E}}

\newcommand{\R}{\mathbb{R}}

\begin{document}
%\linenumbers

\title[Insurance policies with cash flows subject to interest rate changes]
{Life insurance policies with cash flows subject to random interest rate changes}
\author[D. Ba\~{n}os]{David Ba\~{n}os}
\address{D. Ba\~{n}os: Department of Mathematics, University of Oslo, Moltke Moes vei 35, P.O. Box 1053 Blindern, 0316 Oslo, Norway.}
\email{davidru@math.uio.no}
 
\maketitle

%\begin{center}
%This Version : February 4th, 2014
%\end{center}

\begin{abstract}
The main purpose of this work is to derive a partial differential equation for the reserves of life insurance liabilities subject to stochastic interest rates where the benefits and premiums depend directly on changes in the interest rate curve. In particular, we allow the payment streams to depend on the performance of an overnight technical interest rate, making them stochastic as well. This opens up for considering new types of contracts based on the performance of the insurer's returns on their own investments. We provide explicit solutions for the reserves when the premiums and benefits vary according to interest rate levels or averages under the Vasicek model and conduct some simulations computing reserve surfaces numerically. We also give an example of a reinsurance treaty taking over pension payments when the insurer's average returns fall under some specified threshold.
\end{abstract}
 


\section{Introduction}

In actuarial science, a reserve is a liability equal to the present value of future cash flows that the insurance company promises to pay out to the insured under certain conditions. In easier terms, a reserve is an estimate of \emph{how much} the insurance company should charge today in order to meet future payments owed to the insured. This quantity is the basis to ensure solvency of the company and the standard way of computing premiums. The literature sometimes distinguishes between the terms \emph{present value} of future obligations and \emph{actuarial reserve}. The former is the cost of the insurance itself, while the latter is obtained by subtracting the premiums provided by the insured, which are used to pay back to the insured. In this sense, the paid-in premiums should match, in average, the future obligations in such a way that the reserve at the entry of the contract is null.

In life insurance, the main sources of risk of a policy are the state of the insured which triggers the payments, and the future development of the return on the financial investments or technical interest rate. In this note, we focus on the latter and model its risk via a continuous time stochastic process of Itô-diffusion type and Markovian. We suggest a way to make cash flows interest rate dependent and provide some new examples to show how this may relax the burden of low interest rate regimes.

Life insurance claims in the context of stochastic interest rates has been studied before. We mention \cite{Bacinello93, Bacinello94} for some specific interest rate models and \cite{Bacinello, Kurtz96, Nielsen95} for more general models in the framework of Heath-Jarrow-Morton. A model for the interest rate of diffusion type was considered by Norberg and Møller in \cite{Norberg96} and by the authors in \cite{Banos20} where they look at unit-linked insurances with variance risk, as well.

A classical way of computing reserves in the continuous time setting is by solving the so-called Thiele's differential equation, which, in the case of deterministic interest rate is an ordinary differential equation and, in the case of stochastic interest rate, a partial differential equation (PDE). The corresponding Thiele's equation for the case of stochastic rates was derived by Norberg and Møller in \cite{Norberg96} and later risk adjusted by Persson in \cite{Persson98}. In the present manuscript we use the \emph{no arbitrage} approach as in \cite{Persson98} to price insurance claims. A justification of why this is the right pricing approach can be found there. In addition, a complete and thourough discussion of the no arbitrage approach to Thiele's equation is creditted to Steffensen in \cite{Steffensen00}. While interest rate there was taken to be deterministic, the reader may benefit from a precise and excellent discussion on the topic.

Although the above-mentioned works do consider stochastic rates, their models assume that payments are deterministic given the state of the policyholders stipulated by contract. In this note we want to generalize the type of contracts to those looking at interest rate behaviour both at punctual times and path-dependent past averages. In the spirit of \cite{Persson98} we generalize Thiele's partial differential equation (Thiele's PDE) to incorporate interest rate dependent cash flows. The motivation to do so comes from the fact that insurance policies have historically considered fixed returns for their policyholders. This has proven to carry the risk of low returns in the bond market and as a consequence, limitations in order to cover the agreed future liabilities. In this regard, we propose new policies that take into account the behaviour of the returns and how to price them. We also propose a reinsurance treaty which covers a percentage of the pension policies in case of low returns.

The paper is organized as follows: In Section \ref{framework} we introduce our modelling framekwork for the insurance and the interest rate models in two independent probability spaces coupled together. In Section \ref{section:reserve} we review the classical pricing of financial derivatives and adapt the Feynman-Kac formula to both European and Asian-type options, where the latter are options whose contingent claim is a function of the integral of the underlying. Then we derive a Thiele partial differential equation for insurance contracts whose cash flows depend on the performance of the interest rate. We look at specific regime switching contracts in Section \ref{section:stochres}. The examples are given in general terms, but the simulations and analysis of concrete results are provided under the Vasicek interest rate model using a rather simplified model for Norwegian mortality. More specifically we look at: a pure endowment with reduction on premiums under high interest rate levels, pension plans with pension rise under high interest rate regimes, interest rate caps and floors insurances, a binary endowment associated to two average interest rate levels and finally a reinsurance treaty taking over part of the pensions if the average returns during the premium phase are below some threshold.

\section{Framework}\label{framework}

Our modelling framework will consist of two independent probability spaces. On the one hand, the states of the insured at any given time, which will be modelled by a Markov chain and on the other hand, the value of an overnight technical interest rate.

\subsection{The insurance model}

Let $(\Omega_1, \mathcal{A}_1,\PP_1)$ be a complete probability space where $\Omega_1$ is the set of outcomes, $\mathcal{A}_1$ is a suitable $\sigma$-algebra and $\PP_1$ is a probability measure on $(\Omega_1,\mathcal{A}_1)$ and we denote by $\E_1$ the expectation under $\PP_1$. This space carries a stochastic process $X=\{X_t, t\in [0,T]\}$ where $T>0$ is a time horizon, possibly infinite. We will consider the ($\PP_1$-augmented) filtration generated by $X$ and denote it by $\mathcal{F}^X=\{\mathcal{F}_t^X\}_{t\in [0,T]}$. Since we consider no further information on this space, it is natural to take $\mathcal{A}_1=\mathcal{F}_T^X$. We assume that $X$ is a regular continuous time Markov chain taking values in a finite state space $\mathscr{S}$ with transition probabilities
$$p_{ij}(s,t) \triangleq \PP_1 \left[ X_t = j| X_s=i\right] , \quad i,j\in \mathscr{S}.$$

By regular here, we mean that the transition rates
$$\mu_{ij}(t) \triangleq \lim_{\substack{h\to 0 \\ h>0}} \frac{p_{ij}(t,t+h)}{h}, \quad i,j \in \mathscr{S},\quad i\neq j$$
and
$$\mu_{i}(t) \triangleq \lim_{\substack{h\to 0 \\ h>0}} \frac{1-p_{ii}(t,t+h)}{h}, \quad i \in \mathscr{S}$$
exist for every $t\in [0,T]$, are finite and the functions $\mu_{ij}$, $i,j\in \mathscr{S}$ are continuous. Hence, we can recover $\{p_{ij}\}_{i,j\in \mathscr{S}}$ from $\{\mu_{ij}\}_{i,j\in \mathscr{S}}$ via Kolmogorov's equations.

On this space we inherently have the following processes
$$I_i^X(t) = \begin{cases} 1, \mbox{ if } X_t=i,\\ 0, \mbox{ if } X_t\neq i \end{cases},$$
$$N_{ij}^X(t) = \# \{s\in (0,t): X_{s^-} = i, X_s=j \}.$$
Here, $\#$ denotes the counting measure and $X_{t^-}:= \lim_{\substack{u\to t\\ u<t}}X_u$ the left a.s.-limit of $X$ at the point $t$. The random variable $I_i^X(t)$ tells us whether the insured is in state $i$ at time $t$ and $N_{ij}^X(t)$ tells us the number of transitions from $i$ to $j$ in the whole period $(0,t)$.

\begin{defi}[Stochastic cash flow]
A stochastic cash flow is a stochastic process $A=\{A_t\}_{t\geq 0}$ with almost all sample paths with bounded variation.
\end{defi}

More concretely, we will consider cash flows described by an insurance policy entirely determined by its policy functions. We denote by $a_i(t)$, $i\in \mathscr{S}$, the sum of payments from the insurer to the insured up to time $t$, given that we know that the insured has always been in state $i$. Moreover, we will denote by $a_{ij}(t)$, $i,j \in \mathscr{S}$, $i\neq j$, denotes the payments which are due when the insured switches state from $i$ to $j$ at time $t$. We always assume that these functions are of bounded variation (almost surely in the case of random payments). The cash flows we will consider are entirely described by the policy functions defined by an insurance policy.

\begin{defi}[Policy cash flow]
We consider payout functions $a_i(t)$, $i\in S$ and $a_{ij}(t)$, $i,j\in S$, $i\neq j$ for $t\geq 0$ of bounded variation. The (stochastic) cash flow associated to this insurance is defined by
$$A(t) = \sum_{i\in S} A_i(t) + \sum_{\substack{i,j\in S\\ i\neq j}} A_{ij}(t),$$
where
$$A_i(t) = \int_0^t I_i^X(s) da_i(s), \quad A_{ij}(t) = \int_0^t a_{ij}(s) dN_{ij}^X(s).$$

The quantity $A_i$ corresponds to the accumulated liabilities while the insured is in state $i$ and $A_{ij}$ for the case when the insured switches from $i$ to $j$.
\end{defi}

The value of a stochastic cash flow $A$ at time $t$ will be denoted by $V(t,A)$ and is defined as
$$V(t,A)\triangleq \frac{1}{v(t)}\int_t^\infty v(s) dA(s),$$
where $v$ is a suitable discount factor (e.g. $v(t)$ can be today's value of one monetary unit at time $t$ with respect to a technical interest rate, see Definition \ref{discount}).

The stochastic integral is a well-defined pathwise Riemann-Stieltjes integral since $A$ is of bounded variation with probability one. Since we wish to have the best forecast for the reserve at a given time $t$, it is natural to project on to the available information of the insurance company which, at this point, is given by $\mathcal{F}_t^X$. Hence, the present value of the cash flow $A$ given information $\mathcal{F}_t^X$ is given by
\begin{align}\label{eq:Vi}
V_{\mathcal{F}^X}^+(t,A)\triangleq \E_1 [V(t,A)|\mathcal{F}_t^X].
\end{align}
In the above expression one can use any type of information in the conditional expectation to obtain proper reserves. In this case however, we may use the Markovianity of $X$ in order to compute $V_{\mathcal{F}^X}^+(t,A)$ as an actual value, that is by the Markov property of $X$,
$$V_{\mathcal{F}^X}^+(t,A) = \E_1 [V(t,A)|\mathcal{F}_t^X] = \E_1 [V(t,A)|\sigma(X_t)]=H(t,X_t)$$
for some Borel-measurable function $H$, and hence
$$V_i^+(t,A) \triangleq H(t,i)$$
which corresponds to \eqref{eq:Vi} and denotes the present value of $A$ given that we know that the insured is in state $i$ at time $t$. An often abuse of notation is to simply write
$$V_i^+(t,A)=\E_1 [V(t,A)|X_t=i].$$

Having established our insurance model, we will next introduce our market model with stochastic interest rate in an independent probability space.

\subsection{The technical interest rate model}

Let $(\Omega_2,\mathcal{A}_2, \PP_2)$ be a complete probability space where $\Omega_2$ is the set of outcomes, $\mathcal{A}_2$ is a suitable $\sigma$-algebra and $\PP_2$ is a probability measure on $(\Omega_2,\mathcal{A}_2)$ and we denote by $\E_2$ the expectation under $\PP_2$. This space carries a stochastic process: $r=\{r_t, t\in [0,T]\}$ modelling the value of an overnight technical interest rate. On this space we consider the ($\PP_2$-augmented) filtration generated by this process, denoted by $\mathcal{F}^{r}=\{\mathcal{F}_t^{r}\}$ and since no further information is considered we take $\mathcal{A}_2=\mathcal{F}_T^{r}$. A natural modeling assumption at this point is to assume that $\mathcal{F}^X$ and $\mathcal{F}^r$ are independent.

More concretely, $r$ will be governed by the following stochastic differential equation of Itô type.

\begin{align}\label{eq:r}
dr_t=\lambda(t,r_t)dt + \tau(t,r_t)dW_t,\quad r_0=x\in \mathbb{R}, \quad t\in [0,T].
\end{align}
Here, $W$ is a standard Wiener process. Moreover, $\lambda,\tau:[0,T]\times \mathbb{R} \rightarrow \mathbb{R}$ are measurable functions such that a pathwise unique global strong solution exists, and all moments are finite. For instance, if they satisfy the Lipschitz property uniformly in time, have at most linear growth and $\tau$ is away from $0$, we have such result.

Moreover, $r$ from \eqref{eq:r} is a strong Markov process admitting a stochastic flow of homeomorphisms satisfying the flow property, i.e $r_t^{s,r_s^{0,x}} = r_t^{0,x}$, where
\begin{align}\label{rflow}
r_t^{s,x} = x+\int_s^t \lambda(u,r_u^{s,x})du + \int_s^t \tau(u,r_u^{s,x})dW_u, \quad 0\leq s\leq t\leq T, \quad x\in \R.
\end{align}

\begin{rem}
Conditions for the existence of unique global strong solutions admitting a stochastic flow of homeomorphisms can be relaxed considerably. Low regularity can be considered if one wishes to, but since this is not the purpose of this work, we impose the classical Itô assumptions.
\end{rem}

\begin{defi}[Discount factor]\label{discount}
We define the (stochastic) discount factor as the process
$$v(t) = e^{-\int_0^t r_s ds}, \quad t\geq 0,$$
which models the current value of one monetary unit at time $t$.
\end{defi}
Hence, a liability $Z$ which is due at time $s>0$ has a present value today of $v(s)Z$, and value $\frac{v(s)}{v(t)}Z$ at time $t$, $0 \leq t \leq s$.

\subsection{The model}

From now on we work on the probability space $(\Omega,\mathcal{A},\PP)$ where $\Omega=\Omega_1\times\Omega_2$, $\mathcal{A}=\mathcal{A}_1\otimes \mathcal{A}_2$ and $\PP=\PP_1\times \PP_2$ and denote by $\E$ the expectation under $\PP$. The overall information considered by the insurer is $\mathcal{F}_t = \mathcal{F}_t^X \otimes \mathcal{F}_t^{r}$, $t\in [0,T]$ and we have for an $\mathcal{F}_T$-measurable random variable $Y$ on $(\Omega, \mathcal{A},\PP)$ that
$$\E [Y| \mathcal{F}_t] = \E_1 [ \E_2 [Y| \mathcal{F}_t^{r}]| \mathcal{F}_t^{X}] =\E_2 [ \E_1 [Y| \mathcal{F}_t^{X}]| \mathcal{F}_t^{r}]$$
for all $t\in [0,T]$.

We will simultaneously work on $\PP$ or restricted to $\PP_i$, $i=1,2$ when necessary, without really distinguishing notations whenever the context is clear. For example, if $Y\in L^1(\Omega,\mathcal{A},\PP)$ and $Y$ only depends on $\omega_1$ then $\E[Y]$ is implicitly both under $\PP$ and $\PP_2$ since $\E[Y]=\int_{\Omega}Y(\omega)\PP(d\omega) =\int_{\Omega_1}Y(\omega_1)\PP_1(d\omega_1) \int_{\Omega_2}\PP_2(d\omega_2)= \E_2[Y]$ since we are working with probability measures.

\subsection{Present value of a policy when $r$, $a_i$ and $a_{ij}$ are deterministic}
We assume for a moment that $v(t)$ is known and $a_i$, $a_{ij}$, $i,j\in \mathscr{S}$ are deterministic. In particular, the only information available for the insurer is the state of the insured. Then \cite[Theorem 4.6.10]{Koller12} states that the present value of a cash flow $A$ given by an insurance policy entirely described by the functions $a_i$ and $a_{ij}$ is given by
\begin{align}\label{reserve1}
V_i^+(t,A)=\sum_{j\in \mathscr{S}}\left[ \int_t^\infty\frac{v(s)}{v(t)}p_{ij}(t,s) da_j(s) + \sum_{\substack{k\in \mathscr{S}\\k\neq j}} \int_t^\infty \frac{v(s)}{v(t)}p_{ij}(t,s)\mu_{jk}(s) a_{jk}(s)ds\right].
\end{align}
The proof of the above formula is a consequence of \cite[Theorem 4.6.3]{Koller12} which corresponds to \eqref{eq:Vi} using the Markov property on the process $X$. The formula itself is quite intuitive. The (future) value of the policy is the sum over all states of discounted accumulated payment streams coming from the insured being or changing state.

If we assume that $a_i$ are a.e. differentiable with derivative $\dot{a}_i$ and with possibly countable discontinuities at say, points $t_1,\dots,t_n \in [0,\infty)$, $n\geq 1$ then we can recast \eqref{reserve1} as
\begin{align}\label{reserve2}
\begin{split}
V_i^+(t,A)=& \, \sum_{j\in \mathscr{S}}\bigg[\sum_{l=1}^n \frac{v(t_l)}{v(t)} p_{ij}(t,t_l)\Delta a_j(t_l) \mathbb{I}_{\{t<t_l\}}+ \int_t^\infty \frac{v(s)}{v(t)}p_{ij}(t,s) \dot{a}_j(s)ds \\
&+ \sum_{\substack{k\in \mathscr{S}\\k\neq j}} \int_t^\infty \frac{v(s)}{v(t)}p_{ij}(t,s)\mu_{jk}(s) a_{jk}(s)ds\bigg],
\end{split}
\end{align}
where $\Delta a_j(t) \triangleq a_j(t)-a_j(t^-)$ for all $t\geq 0$ and $j\in \mathscr{S}$.

This is like saying that we assume that the functions $a_i$ admit a density (in the generalised sense) of the form
\begin{align}\label{da}
da_j (t) = \sum_{l=1}^n \Delta a_j(t_l) \delta_{t_l}(t) dt + \dot{a}_j(t)dt, \quad t\in [0,\infty), \quad j\in \mathscr{S},
\end{align}
where $y\mapsto \delta_x(y)$ is a (generalised) function called the Dirac delta function which has the property that
$$\int_{-\infty}^{\infty} \delta_x(y)f(y)dy = f(x),\quad x\in\mathbb{R}$$
for all locally integrable functions $f$. It can, for instance, be shown that
$$\varphi_{\varepsilon}(y) = \frac{1}{\sqrt{2\pi\varepsilon}} e^{-\frac{1}{2\varepsilon} (y-x)^2 }, \quad \varepsilon>0$$
converges in $L^1(\mathbb{R})$ as $\varepsilon \to 0$ to $\delta_x$.

%\begin{align}\label{reserve2}
%\begin{split}
%V_i^+(t,A)=& \, \sum_{j\in \mathscr{S}}\bigg[\frac{v(T)}{v(t)} p_{ij}(t,T)(a_j(T)-a_j(T^-))+ \int_t^\infty \frac{v(s)}{v(t)}p_{ij}(t,s) \dot{a}_j(s)ds \\
%&+ \sum_{\substack{k\in \mathscr{S}\\k\neq j}} \int_t^\infty \frac{v(s)}{v(t)}p_{ij}(t,s)\mu_{jk}(s) a_{jk}(s)ds\bigg].
%\end{split}
%\end{align}

Formula \eqref{reserve2} in the case of stochastic interest rates has been studied by \cite{Persson98} and the corresponding Thiele's partial defferential equation was derived. Here, in addition, we are interested in the case where $a_i$ and $a_{ij}$ also are stochastic and subject to the performance of $r$. Hence, the measure in \eqref{da} (or rather its density) is also stochastic.

\section{Pricing and reserving life insurance subject to cash flows with stochastic interest rate}\label{section:reserve}

As we can see in \eqref{reserve2}, if we assume that $a_i$, $a_{ij}$ are adapted to $\mathcal{F}^r$, then we need to price claims of the form
$$\Delta a_j(s)=\Delta a_j(r_{\leq s}), \quad, \dot{a}_j(s)=\dot{a}_j(s,r_{\leq s}), \quad a_{jk}(s)=a_{jk}(s,r_{\leq s}),$$
for possibly varying maturities $s\geq t$. Here, $r_{\leq t}$ denotes the whole trajectory of $r$ up to time $t$, and $\Delta a_j,\dot{a}_j, a_{jk}$ are suitable functionals on the space of continuous functions.

We will focus on claims that are functions of European and Asian type options on $r$. For convenience introduce the following notation
\begin{align}\label{asian}
\overline{r}_{s,t}\triangleq \int_s^t r_u du,\quad s,t\in [0,T], \quad s\leq t,
\end{align}
and $\overline{r}_t \triangleq \overline{r}_{0,t}$.

We will consider payoff functions of the form
$$\theta \left(r_s, \overline{r}_s \right),$$
for fixed $s\in [0,T]$ and some suitable function $\theta: \R^2 \rightarrow \R$. The risk neutral pricing approach tells us that if the market is free of arbitrage, there will be a martingale measure $\QQ$ such that the price at time $t$ of an option with payoff $\theta \left(r_T, \overline{r}_T \right)$ is given by
\begin{align}\label{price0}
\pi_t(T) \triangleq \E_\QQ \left[\frac{v(T)}{v(t)} \theta \left(r_T, \overline{r}_T \right)\Big| \mathcal{F}_t\right].
\end{align}

By the Markov property of $r$ and the fact that $\overline{r}_T=\overline{r}_t+\overline{r}_{t,T}$ we may express $\pi_t(T)$ as a function $u$ of the states of $r_t$ and $\overline{r}_t$ at time $t$. That is,
\begin{align}\label{price}
\pi_t(T) =u\left(t,r_t,\overline{r}_t \right)
\end{align}
where
\begin{align*}
u(t,x,y)\triangleq  \E_\QQ \left[e^{-\int_t^T r_s^{t,x}ds}\theta\left(r_T^{t,x},y+\int_t^T r_s^{t,x}ds\right) \right],
\end{align*}
and $r_{\cdot}^{t,x}$ denotes the process given in \eqref{rflow}.

The pricing measure $\QQ$ is given by
\begin{align*}
\frac{d\QQ}{d\PP}\Big|_{\mathcal{F}_t} = Z_t,\quad t\in [0,T],
\end{align*}
where
$$Z_t\triangleq \mathcal{E} \left( \int_0^{\cdot} \gamma_s dW_s\right)_t,\quad t\in [0,T],$$
and where $\mathcal{E}(Z)_t \triangleq \exp (Z_t-\frac{1}{2}[Z,Z]_t)$ for an Itô process $Z$ and $\gamma$ is an adapted process in $L^2([0,T]\times \Omega)$.

Girsanov's theorem implies that the process $W^{\QQ}$ defined by
$$W_t^{\QQ}\triangleq W_t -\int_0^t \gamma_s ds, \quad t\in [0,T],$$
is a $\QQ$-Wiener processes.

More concretely, assuming that $\gamma_t=\gamma(t,r_t)$ we have that the $\QQ$-dynamics of $r$ are given by
\begin{align*}
dr_t=(\lambda(t,r_t)+\tau(t,r_t)\gamma(t,r_t) )dt + \tau(t,r_t)dW_t^{\QQ}, \quad r_0=x, \quad t\in [0,T]. 
\end{align*}

The following result is the celebrated Feynman-Kac formula which can be used to find the value of \eqref{price} using the theory of partial differential equations. The classical formula is often offered for European type options. Here, we provide a version for path-dependent options adapted to our purposes.

\begin{thm}[Feynman-Kac formula]\label{Feynman-Kac}
Let $Y_t$ be an Itô diffusion given by
$$dY_t = b(t,Y_t)dt + a(t,Y_t)dB_t, \quad Y_0\in \mathbb{R}, \quad t\in [0,T],$$
where $b:[0,T]\times \mathbb{R} \rightarrow \mathbb{R}$, $\sigma:[0,T]\times \mathbb{R} \rightarrow \mathbb{R}$ satisfy classical Itô assumptions for existence and uniqueness and $B$ is a standard Wiener process.

Let $R,S:[0,T]\times\R\rightarrow\R$ be two continuous functions. Consider the function $(t,x,y)\mapsto U(t,x,y)$ solving the following PDE
\begin{align}\label{FKPDE}
\partial_t U +b(t,x) \partial_{x} U +x\partial_yU+ \frac{1}{2} a^2 (t,x) \partial_{x}^2 U -R(t,x)U = 0,
\end{align}
with terminal condition $U(T,x,y)=\psi(t,x,y)$.

Then
$$U(t,x,y)=\E\left[e^{-\int_t^T R(s,Y_s)ds}\psi\left(Y_T,\int_0^T Y_sds\right) \Big| Y_t = x, \int_0^tY_sds=y \right]$$
for all $(t,x)\in [0,T]\times \R$.
\end{thm}
\begin{proof}
Define the process
\begin{align}\label{procZ}
Z_s \triangleq e^{-\int_t^s R(v,Y_v)dv} U\left(s,Y_s, \int_0^sY_vdv \right),\quad t\leq s\leq T.
\end{align}

We denote by $\partial_s U$, $\partial_xU$ and $\partial_yU$ the partial derivaties of $(s,x,y)\mapsto U(s,x,y)$ with respect to $s$, $x$ and $y$, respectively. Itô's formula yields
\begin{align*}
dZ_s =&\, -R(s,Y_s) e^{-\int_t^s R(v,Y_v)dv} U ds + e^{-\int_t^s R(v,Y_v)dv} dU\\
&\hspace{-1cm}= e^{-\int_t^s R(v,Y_v)dv}\left[-R(s,Y_s) U+\partial_s U + b(s,Y_s)\partial_x U + Y_s\partial_y U + \frac{1}{2} a(s,Y_s)^2 \partial_x^2 U\right]ds\\
&+ e^{-\int_t^s R(v,Y_v)du} a(s,Y_s)\partial_x U dB_s
\end{align*}

The finite variation part is $0$ since $U$ satisfies PDE \eqref{FKPDE}. Hence,
$$Z_T-Z_t = \int_t^Te^{-\int_t^s R(v,Y_v)dv} a(s,Y_s)\partial_x U\left(s,Y_s,\int_0^s Y_vdv\right)  dB_s.$$
Taking expectations and using the martingale property of the Itô integral we have
\begin{align}\label{rel1}
\E\left[Z_T\Big|Y_t=x,\int_0^t Y_sds=y\right]=\E\left[Z_t\Big|Y_t=x,\int_0^tY_sds=y\right].
\end{align}

Due to \eqref{procZ} we have
$$\E\left[Z_t\Big|Y_t=x,\int_0^t Y_sds=y\right] = U(t,x,y).$$

Hence, by \eqref{rel1} and the fact that $U(T,x,y)=\psi(x,y)$ we conclude that
$$U(t,x,y)=\E\left[e^{-\int_t^T R(v,Y_v)dv}\psi\left(Y_T, \int_0^T Y_sds \right)\Big|Y_t=x,\int_0^tY_sds=y\right].$$

\end{proof}
%\begin{thm}[Feynman-Kac formula]\label{Feynman-Kac}
%Let $Y_t$ be a $d$-dimensional Itô diffusion given by
%$$dY_t = b(t,Y_t)dt + a(t,Y_t)dB_t, \quad Y_0\in \mathbb{R}^d, \quad t\in [0,T],$$
%where $b:[0,T]\times \mathbb{R}^d \rightarrow \mathbb{R}^d$, $\sigma:[0,T]\times \mathbb{R}^d \rightarrow \mathbb{R}^{d\times m}$ satisfy classical Itô assumptions for existence and uniqueness and $B$ is an $m$-dimensional (uncorrelated) standard Wiener process. Let $L$ denote the differential operator acting on $f\in C^{1,2}([0,T]\times \mathbb{R}^d)$ defined as
%$$Lf = \sum_{i=1}^d b(t,x) \partial_{x_i} f + \frac{1}{2} \sum_{i=1}^d (a\cdot a^T)_{ij} (t,x) \partial_{x_i}\partial_{x_j} f$$
%where $(a\cdot a^T)_{ij} (t,x)$ denotes the $(i,j)$ entry of the matrix $a\cdot a^T$ being $a^T$ the transpose of $a$.
%
%Let $u\in C^{1,2}([0,T]\times \mathbb{R}^d)$ and $R\in C([0,T]\times \mathbb{R}^d)$ satisfy
%$$\partial_t u + Lu -Ru = 0$$
%with final condition $u(T,x)=\psi(x)$, where $\psi:\R^d\rightarrow \R$. Then
%$$u(t,x)=\E\left[e^{-\int_t^T R(s,Y_s)ds}\psi(Y_T) \Big| Y_t = x\right]$$
%for all $(t,x)\in [0,T]\times \R^d$.
%\end{thm}

Applying the above result to the case $Y_t=r_t$, $R(t,x)=x$ we have
$$b(t,x)=\lambda(t,x) +\gamma(t,x)\tau(t,x), \quad a(t,x)= \tau(t,x).$$

Hence, $u$ from \eqref{price} solves the following PDE
\begin{align}\label{PDE1}
\partial_t u+(\lambda(t,x)+\gamma(t,x)\tau(t,x))\partial_x u+x\partial_yu+\frac{1}{2} \tau(t,x)^2 \partial_x^2 u- xu =0
\end{align}
with final condition $u(T,x,y)=\theta(x,y)$.

Denote the differential operator
\begin{align}\label{DiffOp}
\begin{split}
L =& \,(\lambda(t,x)+\gamma(t,x)\tau(t,x))\partial_x +x\partial_y+\frac{1}{2} \tau(t,x)^2 \partial_x^2 .
\end{split}
\end{align}

We will denote by $u_T^\theta$ the solution to the PDE \eqref{PDE1} with payoff function $\theta$ and maturity time $T$. That is
\begin{align}\label{PDEu}
\partial_t u_T^\theta (t,x,y) + Lu_T^\theta (t,x,y) = xu_T^\theta (t,x,y),\quad u_T^\theta(T,x,y)=\theta(x,y).
\end{align}

From now on, we will assume that the policy functions $a_j$ are $\mathbb{P}$-a.s. a.e. differentiable with (stochastic) Lebesgue-Stieltjes measure $da_j$ given by
\begin{align}\label{da2}
da_j(t) = \sum_{l=1}^n \Delta a_j(t_l) \delta_{t_l}(t) dt +\dot{a}_j(t) dt, \quad t\in [0,\infty), \quad j\in \mathscr{S}, 
\end{align}
with
\begin{align}\label{aj}
\Delta a_j(t_l) = f_j\left(r_{t_l},\overline{r}_{t_l} \right), \quad \dot{a}_j(t) = g_j\left(t,r_t, \overline{r}_t\right), \quad t\in [0,\infty), \quad j\in \mathscr{S},
\end{align}
and that
\begin{align}\label{ajk}
a_{jk}(t) = h_{jk}\left(t,r_t, \overline{r}_t\right), \quad t\in [0,\infty), \quad j,k\in \mathscr{S}, 
\end{align}
for some functions $f_j:\R^2\rightarrow \R$, $g_j,h_{jk}:[0,T]\times\R^2\rightarrow \R$, $j\in \mathscr{S}$, $k\neq j$ where we recall that $\overline{r}_t$ is the notation in \eqref{asian}. These functions describe the stream of payments for each state of the insured and for the transitions between states subject to interest rate instantaneous and average changes.

In general, the cash flows coming from a policy are stochastic due to the fact that the states of the insured are stochastic, while the payments are deterministic, given that the state is known. Here, the payments are also stochastic \emph{per se}, since they also depend on the interest rate curve.

Using the notation from \eqref{da2} and \eqref{ajk} in connection with \eqref{price} and \eqref{reserve2} we have that the prospective reserve is given by,
\begin{align*}
\begin{split}
V_{i,\mathcal{F}}^+ (t,A)=&\, \sum_{j\in \mathscr{S}} \sum_{l=1}^n p_{ij}(t,t_l)u_{t_l}^{f_j}(t,r_t,\overline{r}_t) \mathbb{I}_{\{t<t_l\}}+\sum_{j\in \mathscr{S}}\int_t^\infty p_{ij}(t,s)u_s^{g_j(s,\cdot)}(t,r_t,\overline{r}_t)ds\\
&+\sum_{\substack{j,k\in \mathscr{S}\\k\neq j}} \int_t^\infty p_{ij}(t,s)\mu_{jk}(s)u_s^{h_{jk}(s,\cdot)}(t,r_t,\overline{r}_t)ds.
\end{split}
\end{align*}
'
In particular, we see that the prospective reserve of interest-rate-linked policies with stochastic interest rate can be expressed as a function $V(t,r_t,\overline{r}_t)$ of $t$, $r_t$ and $\overline{r}_t$. This is due to the fact that the process $\overline{r}$ is $\mathcal{F}$-adapted and $r$ is Markovian. We can characterize the function $V_i(t,x,y)$, $(t,x,y)\in [0,T]\times \R^2$ by deriving the so-called Thiele's (partial) differential equation.

From now on, and without loss of generality, we assume $n=1$ and $t_1=T>0$. Thus
\begin{align}\label{reserveT}
\begin{split}
V_{i,\mathcal{F}}^+ (t,A)=&\, \sum_{j\in \mathscr{S}} p_{ij}(t,T)u_{T}^{f_j}(t,r_t,\overline{r}_t) +\sum_{j\in \mathscr{S}}\int_t^T p_{ij}(t,s)u_s^{g_j(s,\cdot)}(t,r_t,\overline{r}_t)ds\\
&+\sum_{\substack{j,k\in \mathscr{S}\\k\neq j}} \int_t^T p_{ij}(t,s)\mu_{jk}(s)u_s^{h_{jk}(s,\cdot)}(t,r_t,\overline{r}_t)ds.
\end{split}
\end{align}
where $u$ above satisfies the PDE \eqref{PDE1} with the corresponding maturity times and terminal conditions.

In the above expression the first term corresponds to the benefits associated to being in state $j$ at the end of the contract, the second term corresponds to inflow and outflow of benefits and premiums and the last term are benefits from transitions between $j$ to $k$. We see that at the end of the contract we have indeed $V_{i,\mathcal{F}}^+(T,A)=f_i(r_T,\overline{r}_T)$. The following theorem is the corresponding Thiele's parial differential equation for \eqref{reserveT}, that is policies with payment streams subject to both sudden and average changes in the interest rate curve.

\begin{thm}[Thiele's partial differential equation]\label{ThielePDE}
Let $A$ be the payout function determined by the policy functions $f_i$, $g_i$, $i\in \mathscr{S}$ and $h_{ij}$, $i,j\in \mathscr{S}$, $j\neq i$ as defined in \eqref{aj} and \eqref{ajk}. Denote by $V_{i,\mathcal{F}}^+(t,A)$, $t\in [0,T]$ the value of an insurance contract at time $t$ given that the insured is in state $i\in \mathscr{S}$ at time $t$ and given the information $\mathcal{F}_t$. Then $$V_{i,\mathcal{F}}^+(t,A)=V_i(t,r_t,\overline{r}_t),$$
where the function $V_i: [0,T]\times\mathbb{R}^2\rightarrow \mathbb{R}$ is the solution to the following partial differential equation
\begin{align}\label{thiele}
\begin{split}
\partial_t V_i =&\, x V_i-g_i(t,x,y) -\sum_{j\neq i} \mu_{ij}(t) (h_{ij}(t,x,y) + V_j-V_i)-L V_i
\end{split}
\end{align}
with $L$ being the differential operator defined as
\begin{align}\label{DiffOp}
L  =&\,  (\lambda(t,x)+\gamma(t,x)\tau(t,x))\partial_x+ x\partial_y+ \frac{1}{2}\tau^2(t,x)\partial_x^2.
\end{align}
The boundary condition is given by $V_i(T,x,y)=f_i(x,y)$.
\end{thm}
\begin{proof}
The PDE given in \eqref{thiele} is a second order parabolic linear PDE. It is therefore well-posed and admits a unique solution under the assumption that the coefficients are continuous. See e.g. \cite{Evans10}.

Observe that if $\theta_T^1$ and $\theta_T^2$ are two final conditions for the PDE given in Theorem \ref{Feynman-Kac} and $u^{\theta_T^1}$ and $u^{\theta_T^2}$ denote their solutions then $u^{\theta_T^1+\theta_T^2}$ solves the PDE corresponding to the final condition $\theta_T^1 +\theta_T^2$. This is a trivial consequence of the uniqueness of the PDE and the fact that $u^\theta$ appearing in \eqref{price0} is linear in $\theta$.

Define the function
$$V_i(t,x,y) = G_i^T(t,x,y)+\int_t^T F_i^s (t,x,y)ds, \quad t\in [0,T], \quad i\in \mathscr{S},$$
where
\begin{align*}
G_i^T(t,x,y)\triangleq \sum_{j\in \mathscr{S}} p_{ij}(t,T) u_{T}^{f_j}(t,x,y), \quad t\in [0,T],\quad  i\in \mathscr{S},
\end{align*}
and
\begin{align}\label{Fi}
F_i^s(t,x,y)\triangleq \sum_{j\in \mathscr{S}} p_{ij}(t,s) u_{s}^{\theta_j^s}(t,x,y), \quad t\in [0,T],\quad s\in [t,T], \quad i\in \mathscr{S},
\end{align}
where
$$\theta_j^s(x,y) \triangleq g_j(s,x,y) + \sum_{\substack{j,k\in \mathscr{S} \\ k\neq j}} \mu_{jk}(s) h_{jk}(s,x,y).$$

Then the (stochastic) reserve $V_{i,\mathcal{F}}(t,A)$ with payout function $A$ determined by the policy functions $f_i$, $g_i$ and $h_{ij}$ as defined in \eqref{aj} and \eqref{ajk} is given by
$$V_{i,\mathcal{F}}^+(t,A) = V_i(t,x,y)\Big|_{(x,y)=(r_t,\overline{r}_t)} =G_i^T(t,r_t,\overline{r}_t) + \int_t^T F_i^s(t,r_t,\overline{r}_t)ds.$$

Our arguments will be for fixed $T$ and $s\geq t$ being maturity times. Hence $u_T^{f_j}$ and $u_s^{\theta_j^s}$ are well-defined.

First, observe that by Kolmogorov's backward equation we have
\begin{align}\label{dp}
\partial_t p_{ij}(t,s)= \sum_{\substack{k\in \mathscr{S}\\k\neq i}} \mu_{ik}(t) (p_{ij}(t,s)-p_{kj}(t,s)).
\end{align}

Further, we can find a straightforward relation between the derivatives of $F_i^s$ and those of $u_s^{\theta_j^s}$,
\begin{align*}
\partial_t F_i^s &= \sum_j \sum_{k\neq i} \mu_{ik}(t) (p_{ij}(t,s)-p_{kj}(t,s)) u^{\theta_j^s} + \sum_j p_{ij}(t,s) \partial_t u^{\theta_j^s}\\
&=\sum_{k\neq i} \mu_{ik}(t) (F_i^s-F_k^s) + \sum_j p_{ij}(t,s) \partial_t u^{\theta_j^s},
\end{align*}
where we used relation \eqref{dp} first and \eqref{Fi} thereafter. Moreover, if $L$ is the differential operator of \eqref{DiffOp}, then
$$L F_i^s = \sum_j p_{ij}(t,s) L u_s^{\theta_j^s},$$
since $L$ is linear.

For easier readability we drop the point $(t,r_t,\overline{r}_t)$ in the notation. Now, we compute the Itô differential of $F_i^s(t,r_t,\overline{r}_t)$ under the risk neutral measure $\QQ$ in two ways: first, by direct defition, i.e.
\begin{align}\label{dF1}
dF_i^s = (\partial_t F_i^s + LF_i^s) dt +\tau(t,r_t) \partial_x F_i^s  dW_t^{\QQ},
\end{align}
and now we compute $dF_i^s(t,r_t,\overline{r}_t)$ using the relation \eqref{Fi} in connection with the identity \eqref{dp},
\begin{align*}
\begin{split}
dF_i^s =&\, \left( \sum_{k\neq i} \mu_{ik}(t)(F_i^s-F_k^s) + \sum_j p_{ij}(t,s)\left( \partial_t u_s^{\theta_j^s} +Lu_s^{\theta_j^s}\right) \right) dt+\tau(t,r_t) \partial_x F_i^s  dW_t^{\QQ}.
\end{split}
\end{align*}
Now, we use the fact that $u_s^{\theta_j^s}$ is a solution to $\partial_t u_s^{g_j^s} +Lu_s^{g_j^s} =x u_s^{g_j^s}$. Hence,
\begin{align}\label{dF3}
dF_i^s = \left(\sum_{k\neq i} \mu_{ik}(t)(F_i^s-F_k^s) + r_t F_i^s  \right) dt+ \tau(t,r_t) \partial_x F_i^s  dW_t^{\QQ}.
\end{align}

Equating \eqref{dF1} and \eqref{dF3} we obtain the following PDE in time and space for the function $(t,x,y)\mapsto F_i^s(t,x,y)$ for fixed $s$:
\begin{align}\label{PDEF}
\partial_t F_i^s + LF_i^s = \sum_{k\neq i} \mu_{ik}(t)(F_i^s-F_k^s) + x F_i^s.
\end{align}

On the other hand, recall that $V_i(t,r_t,\overline{r}_t) = V_i(t,x,y)\Big|_{(x,y)=(r_t,\overline{r}_t)}$ and
$$V_i(t,x,y)=G_i^T(t,x,y)+\int_t^T F_i^s(t,x,y) ds.$$
Therefore,
$$\partial_t V_i(t,x.y)=\partial_t G_i^T(t,x,y)+ \int_t^T \partial_t F_i^s(t,x,y)ds - \lim_{\substack{s\to t \\ s>t}}F_i^s(t,x,y),$$
where we used Lebesgue's dominated convergence theorem and the fundamental theorem of calculus. Observe that
\begin{align*}
\lim_{\substack{s\to t \\ s>t}}F_i^s(t,x,y)&=\sum_j p_{ij}(t,t) u_t^{\theta_j^t}(t,x,y) = u_t^{\theta_i^t}(t,x,y)=\theta_i^t(x,y) \\
&= g_i(t,x,y) + \sum_{k\neq i}\mu_{ik}(t) h_{ik}(t,x,y).
\end{align*}

Altogether,
$$\int_t^T \partial_t F_i^s(t,x,y)ds = \partial_t V_i(t,x,y)-\partial_t G_i^T(t,x,y)+ g_i(t,x,y) + \sum_{k\neq i}\mu_{ik}(t) h_{ik}(t,x,y).$$
In particular, evaluating at $(t,x,y)=(t,r_t,\overline{r}_t)$ we have
\begin{align}\label{inttT}
\begin{split}
\int_t^T \partial_t F_i^s(t,r_t,\overline{r}_t)ds\\
&\hspace{-2cm}= \partial_t V_i(t,r_t,\overline{r}_t)-\partial_t G_i^T(t,r_t,\overline{r}_t)+ g_i(t,r_t,\overline{r}_t) + \sum_{k\neq i}\mu_{ik}(t) h_{ik}(t,r_t,\overline{r}_t).
\end{split}
\end{align}

Integrating \eqref{PDEF} with respect to $s$ on the region $[t,T]$, interchanging integration and differentiability, using \eqref{inttT} and $V_i(t,x,y)=G_i^T(t,x,y)+\int_t^T F_i^s(t,x,y)ds$ we obtain
$$\partial_t ( V_i - G_i^T) + g_i + \sum_{k\neq i} \mu_{ik}(t) h_{ik} + L (V_i -G_i^T) = \sum_{k\neq i} \mu_{ik}(t) \left(V_i -V_k + G_k^T - G_i^T \right) +x (V_i - G_i^T).$$

All the terms involving $G_i^T$ will cancel each other. Indeed, observe that
\begin{align*}
\partial_t G_i^T &= \sum_j \partial_t p_{ij}(t,T) u_T^{f_j} + \sum_j p_{ij}(t,T) \partial_t u_T^{f_j}\\
&= \sum_{k\neq i} \mu_{ik}(t) (G_i^T - G_k^T) + \sum_j p_{ij}(t,T) \partial_t u_T^{f_j}.
\end{align*}
Also, 
\begin{align*}
\partial_t G_i^T+LG_i^T &= \sum_{k\neq i} \mu_{ik}(t) (G_i^T - G_k^T) + \sum_{j}p_{ij}(t,T) (\partial_t u_T^{f_j} + Lu_T^{f_j})\\
&=\sum_{k\neq i} \mu_{ik}(t) (G_i^T - G_k^T) +xG_i^T,
\end{align*}
where we used that $\partial_t u_T^{f_j} + Lu_T^{f_j} = x u_T^{f_j}$.

As a result, we obtain $V_i(T,x,y)=f_i(x,y)$ and
$$\partial_t V_i = x V_i - g_i(t,x,y) - \sum_{k\neq i} \mu_{ik}(t) h_{ik}(t,x,y) +\sum_{k\neq i} \mu_{ik}(t) (V_i - V_k) - LV_i$$
and the result follows by grouping together the sums over $k\in \mathscr{S}$, $k\neq i$.
\end{proof}

The following result is a direct consequence of Thiele's PDE obtained in Theorem \ref{ThielePDE} showing that one can aggregate reserves from policies of equal ages and expiration dates.
\begin{cor}
Assume we have $n\geq 1$ policies with the same maturities $T$ and all policyholders enter the contract at the same age. Then the present value of all reserves
$$\overline{V}_i(t,x,y) \triangleq \sum_{k=1}^n V_i^k(t,x,y), \quad i \in \mathscr{S},$$
where $V_i^k(t,x,y)$ corresponds to the reserve of policy $k$ at time $t$ and interest rate state $x$, given that the policy holder is in state $i\in \mathscr{S}$, satisfies the PDE
$$\partial_t \overline{V}_i = x\overline{V}_i - \overline{g}_i - \sum_{j\neq i} \mu_{ij}(t) ( \overline{h}_{ij}+\overline{V}_j-\overline{V}_i) - L \overline{V}_i,$$
with terminal condition $\overline{V}_i(T,x,y)=\overline{f}_i(x,y)$ where
$$\overline{f}_i(x,y) = \sum_{k=1}^n f_i^k(x,y), \quad \overline{g}_i(x,y) = \sum_{k=1}^n g_i^k(x,y), \quad \overline{h}_{ij}(x,y) = \sum_{k=1}^n h_{ij}^k(x,y).$$
Here, $f_i^k$, $g_i^k$ and $h_{ik}^k$ denote the policy functions for the $k$th policyholder.
\end{cor}

\section{Life insurance policies with stochastic policy functions}\label{section:stochres}
Henceforward, we assume that $\mathscr{S}=\{\ast,\dag\}$. Since there are only two states and one of them is absorbing, we drop the notations $\ast$ and $\dag$. That is to say, we simply write $\mu$ to denote the force of mortality, $f$, $g$ and $h$ the policy functions and $V$ the mathematical reserve given that the insured is alive, since otherwise $V_\dag \equiv 0$.

For simulations purposes, we will assume the Gompertz-Makeham law of mortality on $\mu$ given by
$$\mu(t) = \alpha_0+\alpha_1\exp(\alpha_2 t),\quad t\geq 0, \quad \alpha_0,\alpha_1,\alpha_2\in \mathbb{R}.$$
This law of mortality describes the age dynamics of human mortality rather accurately in the age window from
about $30$ to $80$ years of age, which is good enough for our analysis. For this reason, we excluded
the very first and last age groups from the data. We obtain $\hat{\alpha_0}= 0.00127529$, $\hat{\alpha_1}=2.51137\cdot 10^{-6}$ and $\hat{\alpha_2}=0.1271853$ which are the least squares estimates obtained by fitting Norwegian mortality from 2019 (both genders together). See Statistics Norway, table: 05381 for the employed data. Our examples will assume that the age of the insured at the beginning of the contract is fixed to $30$ years old.

We will introduce the following notations
\begin{align}\label{UK}
U_s^K(t,r_t) \triangleq \E_{\QQ}\left[ e^{-\int_t^s r_u du} \mathbb{I}_{\{r_s\geq K \}}|\sigma(r_t)\right],
\end{align}
and
\begin{align}\label{UbarK}
\overline{U}_s^K(t,r_t,\overline{r}_t)\triangleq \E_{\QQ}\left[ e^{-\int_t^s r_u du} \mathbb{I}_{\{\overline{r}_s\geq Ks \}}|\sigma(r_t)\otimes\sigma(\overline{r}_t)\right],
\end{align}
which appear naturally in many common policy specifications with regimes, such as endowment, term insurance, pensions, etc.

For simulation purposes we consider the dynamics of the Vasicek short-rate model, which are given by
\begin{align}\label{vasicek}
dr_t = a(b-r_t)dt +\sigma dW_t, \quad r_0\in \mathbb{R}, \quad t\in [0,T],
\end{align}
where $a,b,\sigma\in \mathbb{R}$, $\sigma>0$ are model parameters. In this case the market price of risk $\gamma(t,x)\equiv \gamma\in \mathbb{R}$ is an additional constant parameter. Since the Vasicek model has the property that it is invariant under change of measure, we simply take $\gamma=0$ and the reader may adjust the interest market price of risk by a modification of $a$ and $b$.

Under the model in \eqref{vasicek}, one can find fairly explicit expressions. For example, PDE \eqref{PDEu} can be solved in closed form for some specific terminal conditions $\theta$.

Introduce the notations
$$\mu_{r_s|r_t}(h|x)\triangleq x e^{-ah}+b(1-e^{-ah}), \quad \mu_{\overline{r}_{t,s}|r_t}(h,x)\triangleq (x-b)\frac{1}{a}(1-e^{-ah})+bh,$$
$$\sigma_{r_s|r_t}^2 (h)\triangleq \frac{\sigma^2}{2a}\left( 1- e^{-2ah}\right),\quad  \sigma_{\overline{r}_{t,s}|r_t}^2 (h)\triangleq \frac{\sigma^2}{a^2}\left[ h-2\frac{1}{a}(1-e^{-ah})+\frac{1}{2a}\left(1-e^{-2ah}\right)\right]$$
and the covariance
$$\sigma_{r_s|r_t, \overline{r}_{t,s}|r_t}(h)\triangleq \frac{\sigma^2}{a}\left[\frac{1}{a}(1-e^{-ah})-\frac{1}{2a}\left(1-e^{-2ah}\right)\right],$$
which gives the correlation function
$$\rho_{r_s|r_t, \overline{r}_{t,s}|r_t}(h)\triangleq \frac{\sigma_{r_s|r_t, \overline{r}_{t,s}|r_t}(h)}{\sigma_{r_s|r_t}(h) \sigma_{\overline{r}_{t,s}|r_t}(h)}.$$
It holds that
$$r_s|r_t \sim N(\mu_{r_s|r_t}(s-t|r_t), \sigma_{r_s|r_t}^2 (s-t)),\quad \overline{r}_{t,s}|r_t\sim N(\mu_{\overline{r}_{t,s}|r_t}(s-t,r_t), \sigma_{\overline{r}_{t,s}|r_t}^2 (s-t)).$$

From the above, one can deduce that
\begin{align*}
\begin{split}
U_s^K(t,r_t)&=\\
&\hspace{-1.5cm}= e^{-\mu_{\overline{r}_s|r_t}(h|x) + \frac{1}{2}\sigma_{\overline{r}_s|r_t}^2 (h) }\Phi\left(-\rho_{r_s|r_t, \overline{r}_{t,s}|r_t}(h)\sigma_{\overline{r}_s|r_t} (h)+\frac{\mu_{r_{t,s}|r_t}(h,x)-K}{\sigma_{r_{t,s}|r_t} (h)} \right)\Bigg|_{(h,x)=(s-t,r_t)},
\end{split}
\end{align*}
where $\Phi$ denotes the distribution function of a standard normally distributed random variable. Similarly, one can show that
\begin{align*}
\begin{split}
\overline{U}_s^K(t,r_t,\overline{r}_t)&=\\\\
&\hspace{-1.5cm}= e^{-\mu_{\overline{r}_s|r_t}(h|x) + \frac{1}{2}\sigma_{\overline{r}_s|r_t}^2 (h) }\Phi\left(-\sigma_{\overline{r}_s|r_t} (h)+\frac{\mu_{\overline{r}_{t,s}|r_t}(h,x)-Ks+y}{\sigma_{\overline{r}_{t,s}|r_t} (h)} \right)\Bigg|_{(h,x,y)=(s-t,r_t, \overline{r}_t)}.
\end{split}
\end{align*}
Observe that when $K=-\infty$ then we obtain the classical price of a zero-coupon bond at time $t$ with maturity $s$ and when $K=\infty$ both expectations are null.

\subsection{Pure endowment with premium reduction on high interest rate levels}

Let $E>0$ be the guaranteed endowment to be paid at the end of the contract $T$ upon survival. Let $\rho\in [0,1]$ be a reduction factor and $K>0$ an interest rate level above which premiums are reduced by a factor of $1-\rho$. Then the value of this contract, premiums taken into account is given by
\begin{align*}
V^\rho(t,r_t)=& \, -\pi^\rho \int_t^T p_{\ast\ast}(t,s)\left(U_s^{-\infty}(t,r_t)-\rho U_s^K(t,r_t)\right)ds   +E p_{\ast\ast}(t,T)U_T^{-\infty}(t,r_t),
\end{align*}
where $U_s^{-\infty}$ is the function given in \eqref{UK}.

We choose $\pi^\rho$ in such a way that $V^\rho(0,r_0)=0$. In this case, the expected difference between reserves is given by
$$\E[V^{\rho}(t,r_t)-V^{0}(t,r_t)] = (\pi^0-\pi^\rho) \int_t^T p_{\ast\ast}(t,s)\E[U_s^{-\infty}(t,r_t)]ds +\pi^\rho \rho \int_t^T \E[U_s^K(t,r_t)]ds.$$

In Figure \ref{fig:2} we show an example of two random interest rate curves and the corresponding reserves for a contract ($\times 1\, 000$) with and without reduction and the difference between such reserves. The parameters for the interest rate model are $r_0=0.03$, $a=0.1$, $b=0.2$, $\gamma=0$, $\sigma=0.01$. The contract pays an endowment of $\$ 100\, 000 $ in $T=10$ years for a person who is $30$ years old today. For the reduction case we take a threshold of $K=0.04$ and a premium reduction of $\rho=20\%$.

\begin{figure}[H]
\centering
  \includegraphics[scale=0.6]{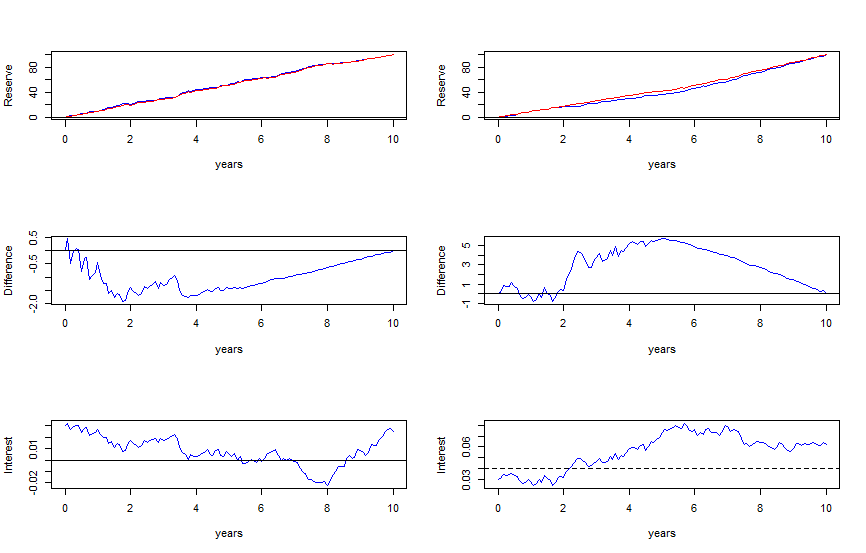}
   \caption{On top: Reserves for an endowment of $\$ 100 \,000$ in $\times 1 \, 000$ units with no reduction (in blue) and with a reduction of $20\%$ on premiums for interest rate above $K=4\%$ (in red). In the middle: difference between reserves. On the bottom: the corresponding interest rate (random) outcomes. The premiums obtained are $\pi^0= \$ 8\, 770.28$ and $\pi^\rho = \$ 9\, 092.40$ when $r_0=3\%$.}
\label{fig:2}
\end{figure}

One may argue that $\pi^\rho = \$ 9\, 092.40$ is a rather high premium for an endowment of $\$ 100 \, 000$, but the insured could have potentially profited from long high interest regimes as, for instance, the outcome on the right. There, the insured has to pay approximately $(1-\rho)\pi^\rho = 7\, 273.92$ for eight years of high interest rate and $\pi^\rho=\$ 9\, 092.40$ for the two first years of the contract. Thus, a total amount of $\$ 76\, 376.16$. Under the same outcome, an insured with no reduction would have paid $10 \pi^0 = \$ 87\, 702.87$. In any case, to overcome the potential issue of paying too high premiums, one can add a premium refund at the end of the contract if the rates are low, or as we will see later, define a policy based on average rates rather than current rates.

We can see in Figure \ref{fig:2} the reserves and their differences for two (random) interest rate regimes; one with regime mostly under $4\%$ (right), and another with regime crossing $K=4\%$ (left). Both outcomes shows that the reserves for the case of a policy with premium reduction requires slightly higher reserves. In Figure \ref{fig:meanendow}, we show the mean difference in the long-run.

\begin{figure}[H]
\centering
\begin{subfigure}{.5\textwidth}
  \centering
  \includegraphics[scale=0.3]{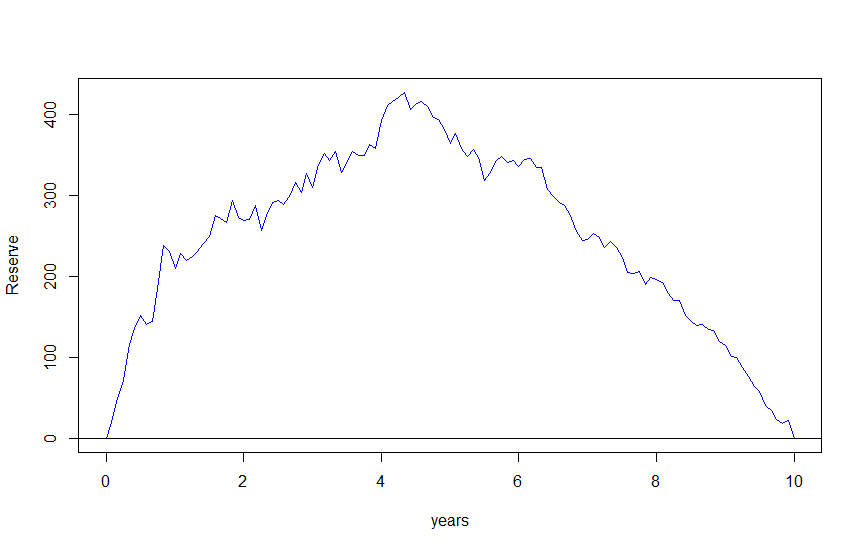}
  \caption{$1\, 000$ simulations ($\sim 16$ min)}
  \label{fig:sub1}
\end{subfigure}%
\begin{subfigure}{.5\textwidth}
  \centering
  \includegraphics[scale=0.3]{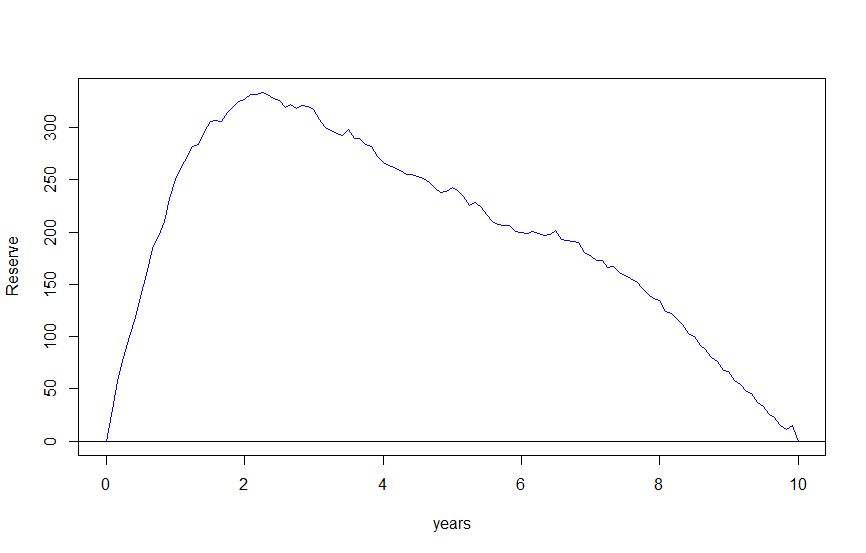}
  \caption{$10\, 000$ simulations ($\sim 2$ h $36$ min)}
  \label{fig:sub2}
\end{subfigure}
\begin{subfigure}{.5\textwidth}
  \centering
  \includegraphics[scale=0.3]{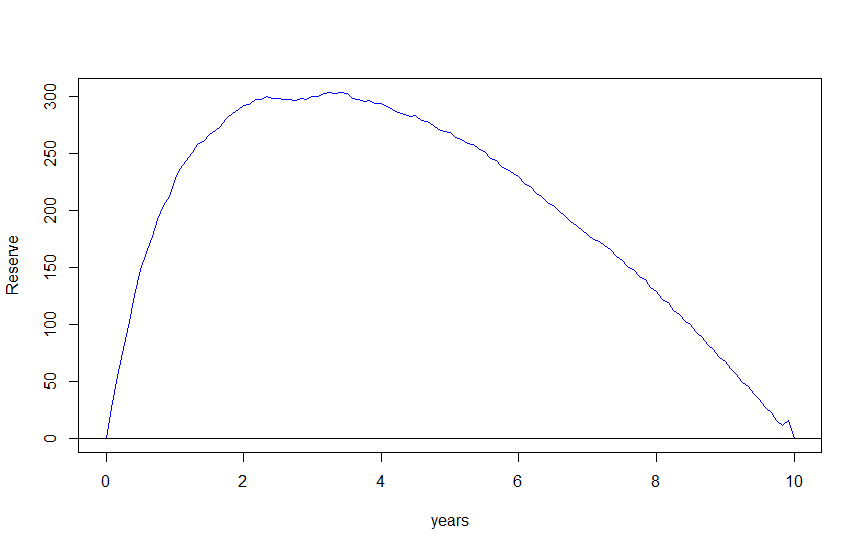}
  \caption{$100\, 000$ simulations ($\sim 1$ d $1$ h $40$ min)}
  \label{fig:sub2}
\end{subfigure}%
\begin{subfigure}{.5\textwidth}
  \centering
  \includegraphics[scale=0.3]{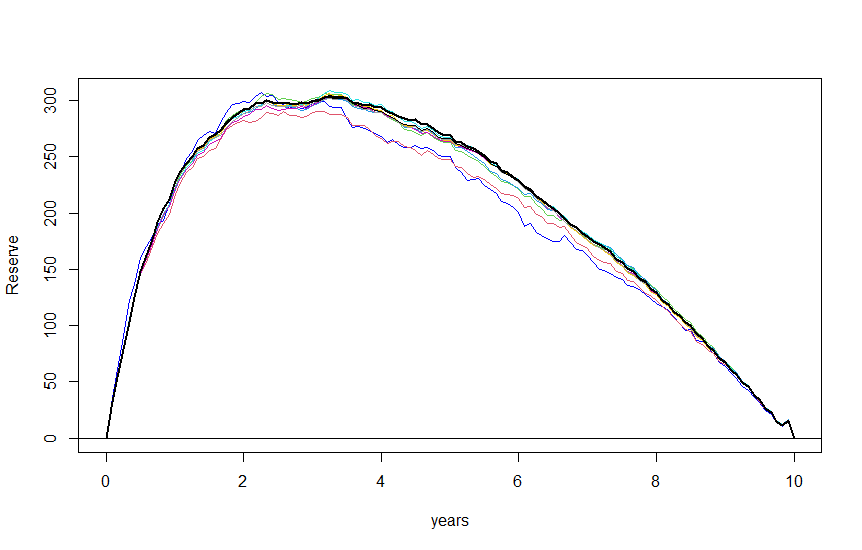}
  \caption{Multiples of $10\, 000$ simulations}
  \label{fig:sub2}
\end{subfigure}
\caption{Mean difference reserve between an endowment with reduction $\rho=20\%$ above $K=4\%$ and no reduction. We can see that the difference is more prominent and significant at the beginning of the contract and that the approximation is more reliable after $10\ 000$ simulations.}
\label{fig:meanendow}
\end{figure}

To finish this example we show the reserve surface, i.e. the function $(t,x)\mapsto V^{\rho}(t,x)$ for $\rho=20\%$ and the surface of the difference between the reserves with reduction and without.
\begin{figure}[H]
\centering
\begin{subfigure}{.5\textwidth}
  \centering
  \includegraphics[scale=0.5]{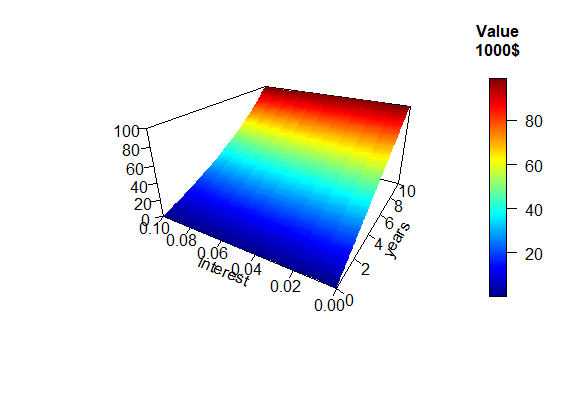}
  \caption{Reserve surface, $\rho=20\%$}
  \label{fig:surface sub1}
\end{subfigure}%
\begin{subfigure}{.5\textwidth}
  \centering
  \includegraphics[scale=0.5]{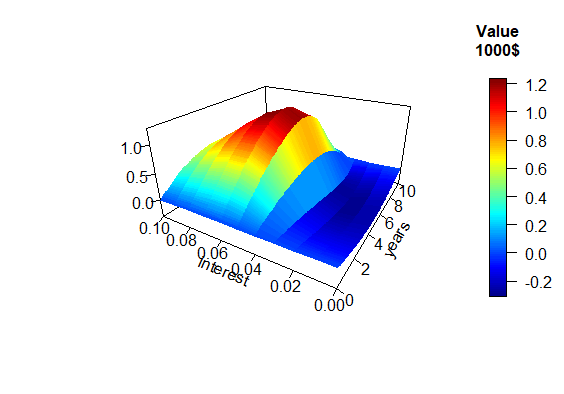}
  \caption{Surface of the difference}
  \label{fig: surface sub2}
\end{subfigure}
\caption{}
\label{fig:surface}
\end{figure}

The surfaces in Figure \ref{fig:surface} gives us complete information of the behaviour of the reserves and their differences for each plausible interest rate value $x$. One can think that all outcomes from Figure \ref{fig:2} are paths along the surface starting at interest $0.03$ and ending at $\$ 100 \, 000$.

%%%%%%%%%%%%%%%%%%%%%%%%%%%%%%%%%%%%%%%%%%%%%%%%%%%%%%%%
%%%%%%%%%%%%%%%%%%%%%%%%%%%%%%%%%%%%%%%%%%%%%%%%%%%%%%%%

\subsection{Pension insurance with pension bonus during high interest rate regimes}\label{pensionex}

In the case of a pension insurance, the present value $V^\rho$ of a policy with periodical premiums $\pi$ paying a periodical pension of $P$ if interest is low and $(1+\rho)P$ if interest is high, from time $\hat{T}$ until end of life is given by
\begin{align*}
V^\rho(t,r_t) =&\, -\pi \int_{\min\{t,\hat{T}\}}^{\hat{T}} p_{\ast\ast}(t,s)U_s^{-\infty}(t,r_t)ds\\
&+P \int_{{\max\{t,\hat{T}\}}}^\infty p_{\ast\ast}(t,s)\left(U_s^{-\infty}(t,r_t)+\rho U_s^{K}(t,r_t) \right)ds.
\end{align*}

The PDE associated to $V^\rho$ in this case is given by
$$\partial_t V = xV +\pi \mathbb{I}_{\{x<K, 0\leq t < \hat{T}\}}-P\left( 1+\rho \mathbb{I}_{\{x\geq K\}} \right)\mathbb{I}_{\{\hat{T}\leq t < \infty\}}  +\mu(t)V - LV,$$
where
$$LV=\left[a(b-x)+\gamma \tau \right]\partial_xV-\frac{1}{2}\tau^2\partial_x^2 V$$
with boundary condition $\displaystyle \lim_{t\to \infty} V(t,x)=0$.

Insurance companies usually set the end of the contract at an age of $120$ or similar. Hence, a computationally more friendly boundary condition for the case that the insured is $30$ years old would be $V(90,x)=0$ being the maximum length of the contract $T=90$ years.

We still consider a person who is $30$ years old today and will retire at the age of $70$, that is in $\hat{T}=40$ years from now. In Figure \ref{fig:4} we show an example of two random interest rate curves and the corresponding reserves for a contract ($\times 1\, 000$) with and without bonus on pensions and the difference between such reserves. The parameters for the interest rate model are $r_0=0.03$, $a=0.1$, $b=0.2$, $\gamma=0$, $\sigma=0.01$. The contract pays a pension $P$ of $\$ 20\, 000 $ yearly if $r_t< 4\%$ and $(1+\rho)P$ if $r_t\geq 4\%$. We take a pension bonus of $\rho=20\%$.

\begin{figure}
\centering
  \includegraphics[scale=0.6]{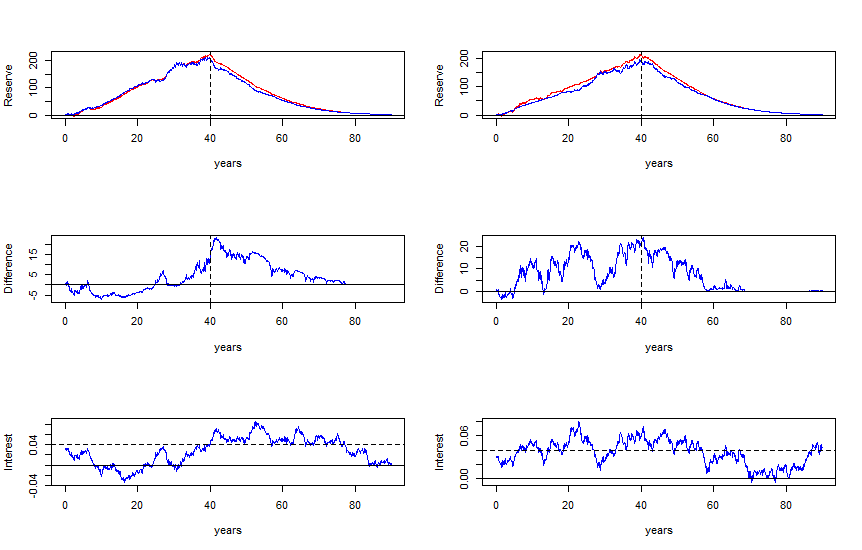}
   \caption{On top: Reserves for a pension insurance of $\$ 20 \,000$ yearly in $\times 1 \, 000$ units with no bonus (in blue) and with a bonus of $20\%$ on pensions for interest rate regimes above $K=4\%$ (in red). In the middle: difference between reserves. On the bottom: the corresponding interest rate (random) outcomes. The premiums obtained are $\pi^0= \$ 8\, 611.31$ and $\pi^\rho = \$ 8\, 910.87$}
\label{fig:4}
\end{figure}

\begin{figure}
\centering
\begin{subfigure}{.5\textwidth}
  \centering
  \includegraphics[scale=0.24]{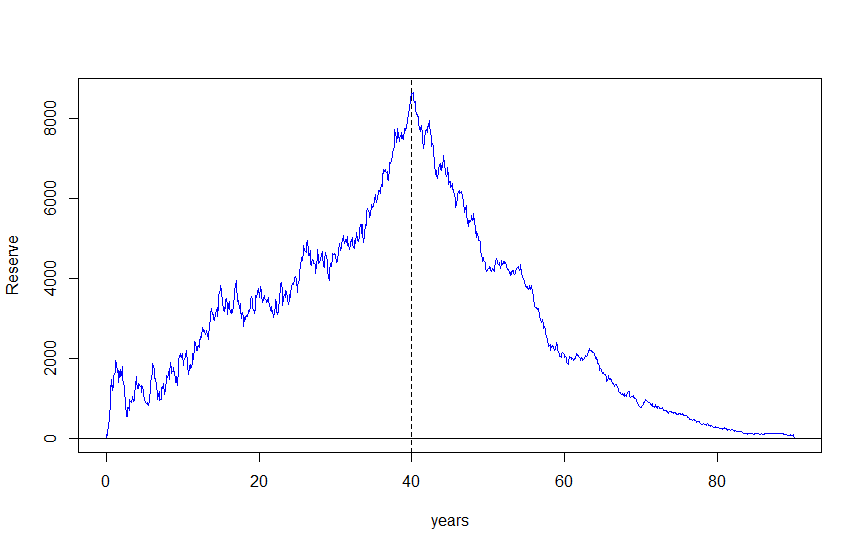}
  \caption{$50$ simulations ($\sim 53$ minutes)}
  \label{fig:sub2}
\end{subfigure}%
\begin{subfigure}{.5\textwidth}
  \centering
  \includegraphics[scale=0.24]{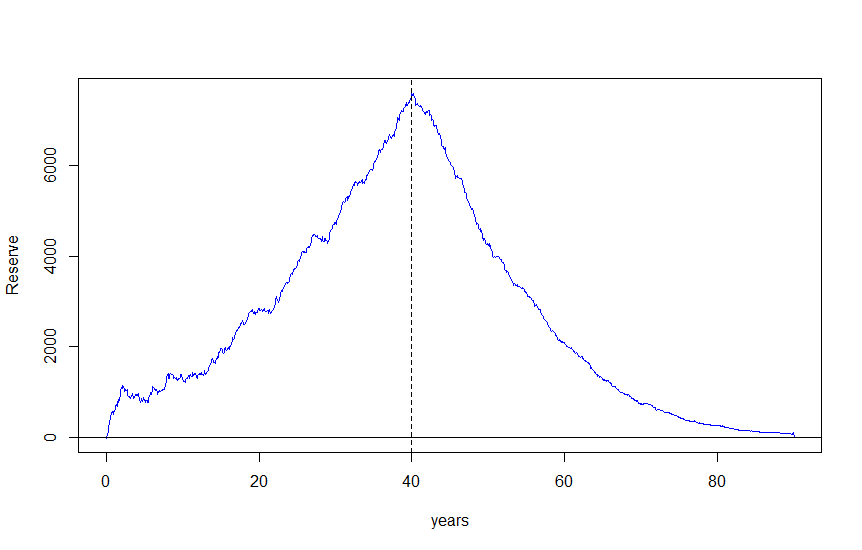}
  \caption{$500$ simulations ($\sim 9$h $31$ minutes)}
  \label{fig:sub2}
\end{subfigure}
\begin{subfigure}{.5\textwidth}
  \centering
  \includegraphics[scale=0.24]{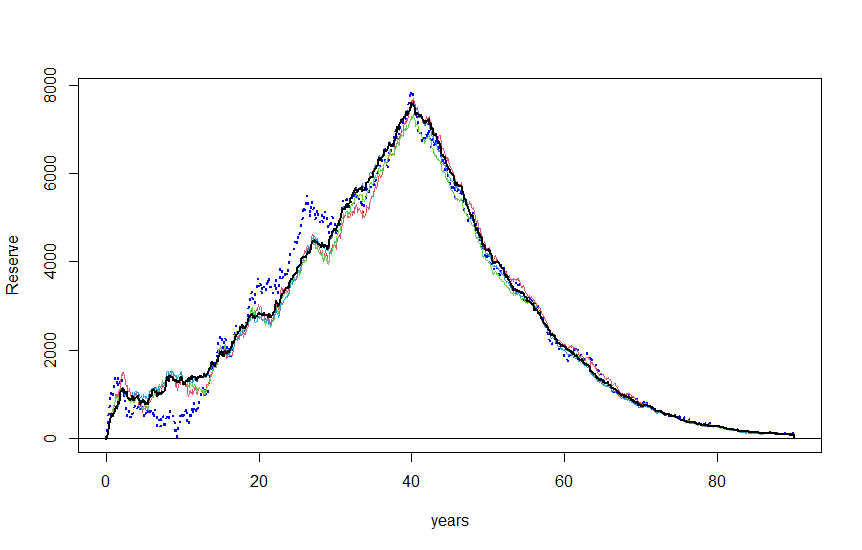}
  \caption{Multiples of $100$ simulations. Dotted line corresponds to $10$ and thick line to $500$.}
  \label{fig:sub2}
\end{subfigure}
\caption{Mean difference reserve between a pension insurance with bonus of $\rho=20\%$ above $K=4\%$ and no bonus. We can see that the difference is more prominent and significant around the retirement age.}
\label{fig:3}
\end{figure}

\begin{figure}[H]
\centering
\begin{subfigure}{.5\textwidth}
  \centering
  \includegraphics[scale=0.32]{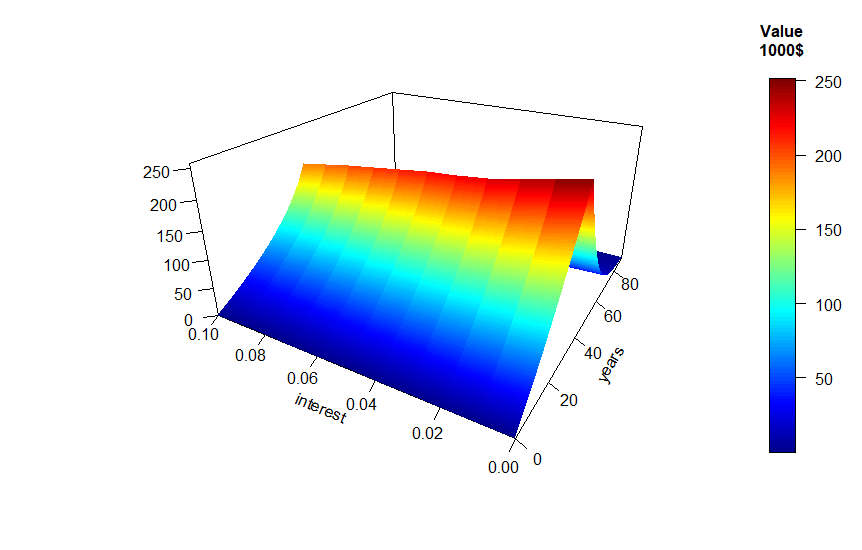}
  \caption{Reserve surface, $\rho=20\%$}
  \label{fig:sub1}
\end{subfigure}%
\begin{subfigure}{.5\textwidth}
  \centering
  \includegraphics[scale=0.32]{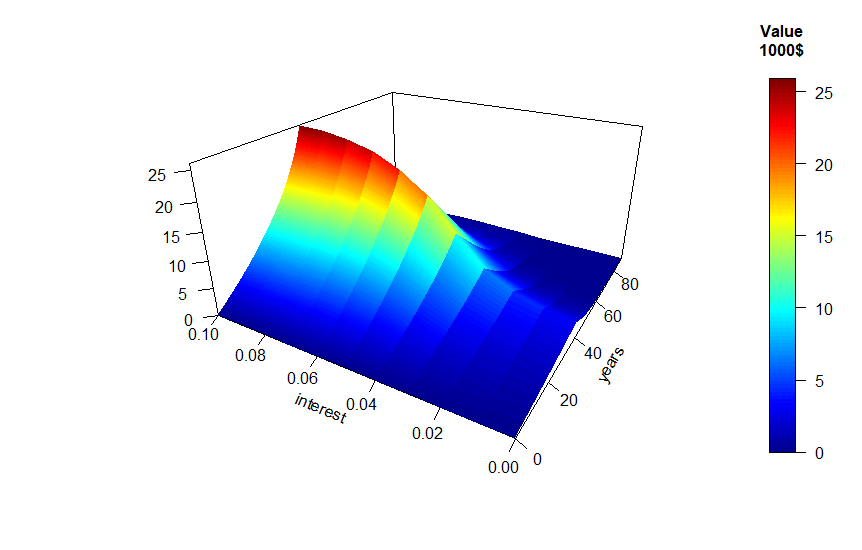}
  \caption{Surface of the difference}
  \label{fig:sub2}
\end{subfigure}
\caption{Both surfaces have been computed using the direct formula for the reserve ($\sim 31$ minutes)}
\label{fig:3}
\end{figure}

\subsection{Interest rate caps and floors insurance}

One can also look at four classical options on $r_T$. Figure \ref{fig:capsfloors} shows the surfaces for the present value of an interest rate cap, an interest rate floor, a call and a put option with strike rate $4\%$ at the end of the contract. We obtain these present values by solving the PDE from Theorem \ref{ThielePDE} which in this case is given by
$$\partial_t V = xV + \mu(t) V -LV,$$
with terminal conditions $V_{\mbox{cap}}(T,x)=E\mathbb{I}_{\{x\geq K\}}$, $V_{\mbox{floor}}(T,x)=E\mathbb{I}_{\{x\leq K\}}$, $V_{\mbox{call}}(T,x)=E\max\{x-K,0\}$ and $V_{\mbox{put}}(T,x)=E\max\{K-x,0\}$. We use an explicit finite difference method to obtain the solutions. The curve we see at $t=0$ corresponds to the fair premium of the contract.

\begin{figure}
\centering
\begin{subfigure}{.5\textwidth}
  \centering
  \includegraphics[scale=0.3]{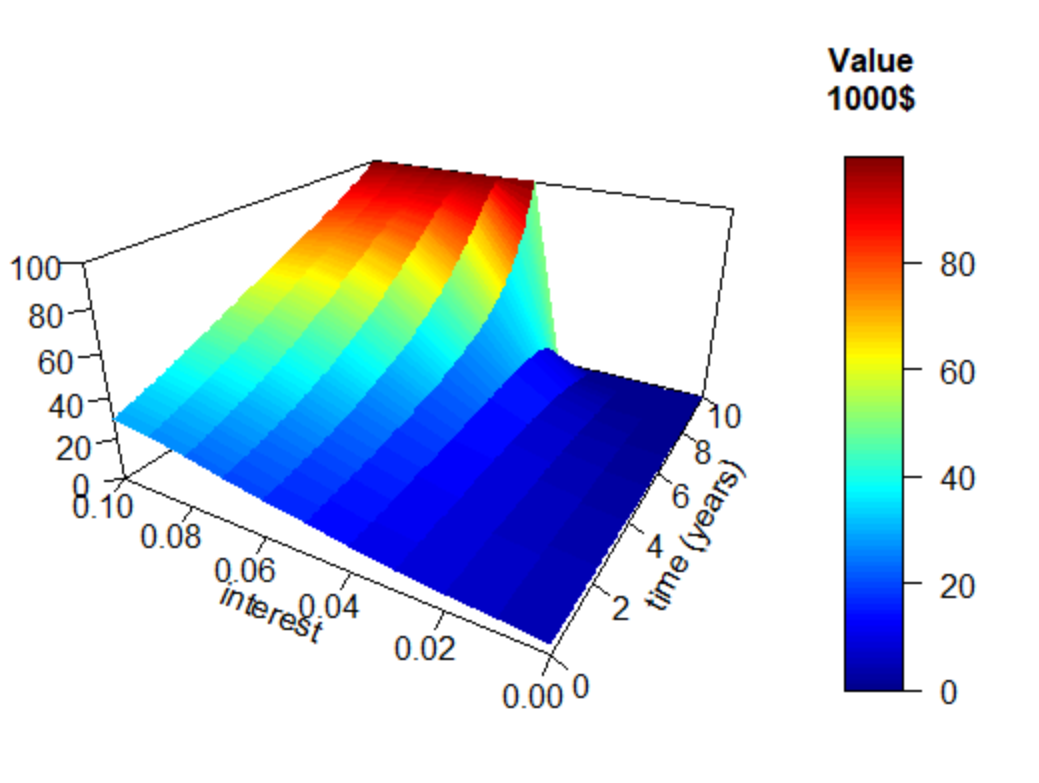}
  \caption{Reserve surface for a cap}
  \label{fig:sub1}
\end{subfigure}%
\begin{subfigure}{.5\textwidth}
  \centering
  \includegraphics[scale=0.3]{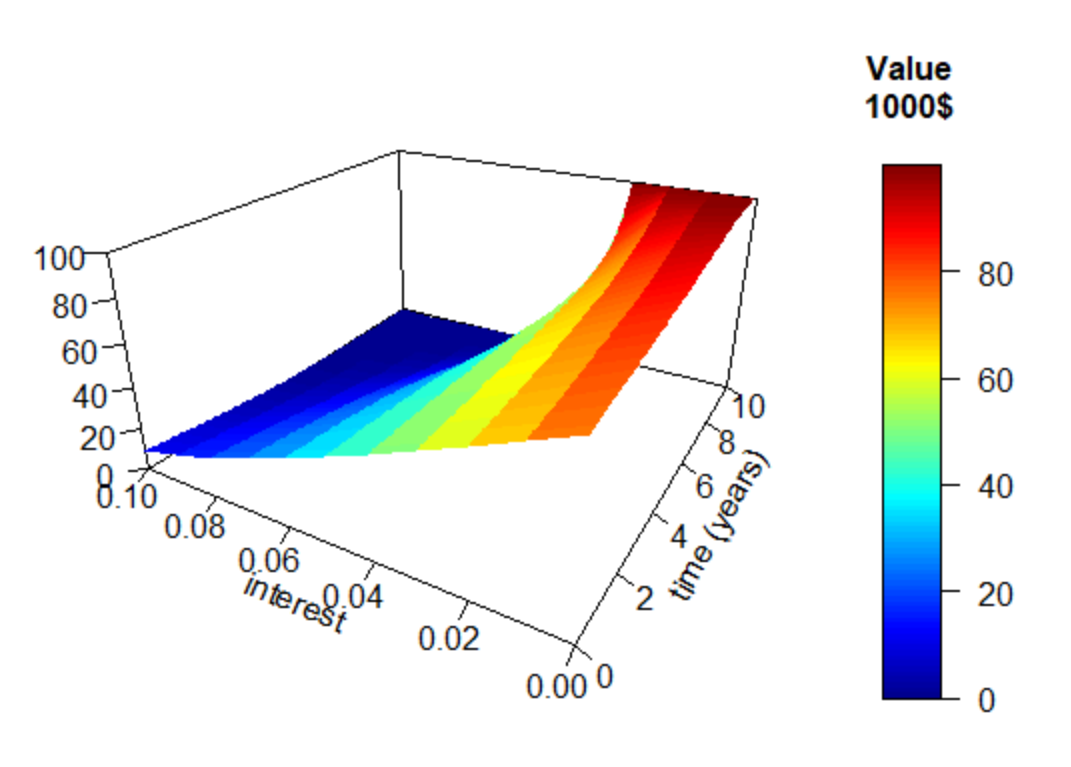}
  \caption{Reserve surface for a floor}
  \label{fig:sub2}
\end{subfigure}
\begin{subfigure}{.5\textwidth}
  \centering
  \includegraphics[scale=0.3]{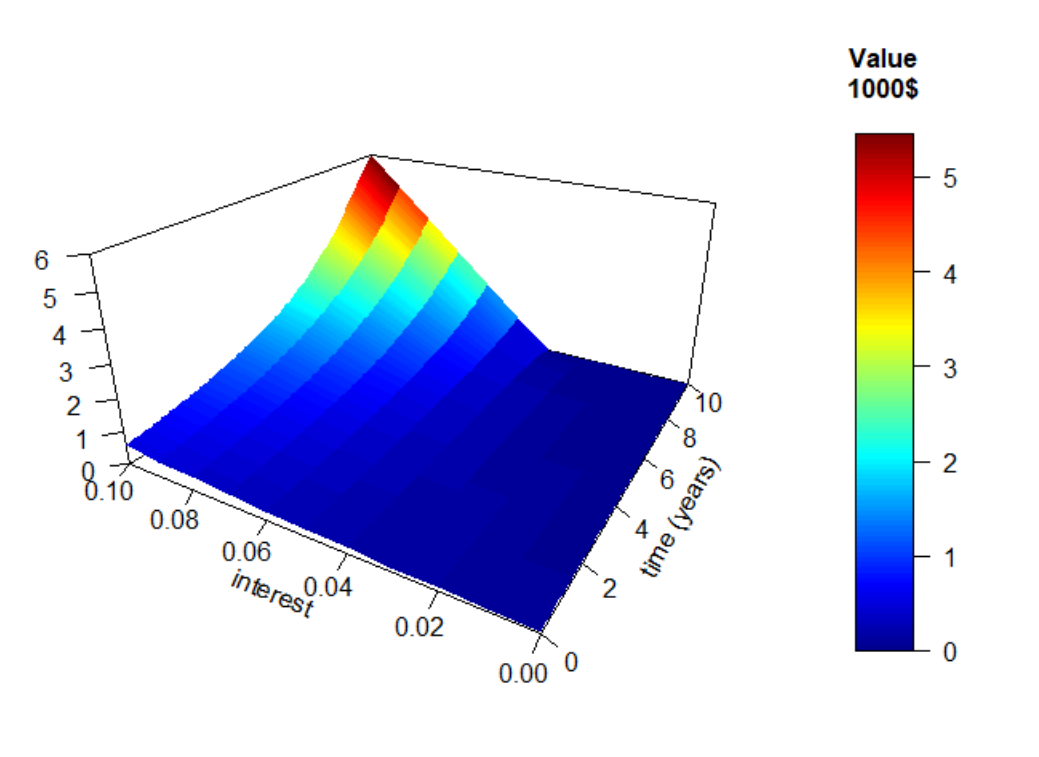}
  \caption{Reserve surface for a caplet}
  \label{fig:sub2}
\end{subfigure}%
\begin{subfigure}{.5\textwidth}
  \centering
  \includegraphics[scale=0.3]{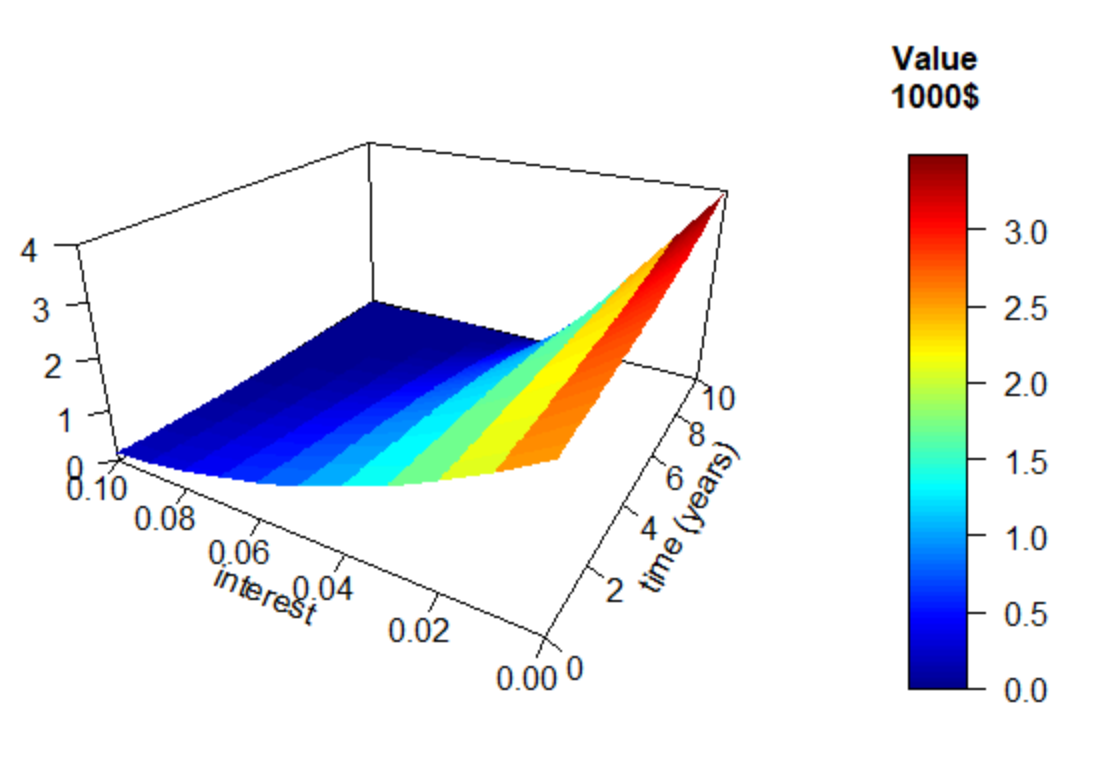}
  \caption{Reserve surface for a floorlet}
  \label{fig:sub2}
\end{subfigure}
\caption{European options on interest rate with strike $K=4\%$. Surface is obtained solving Thiele's PDE \eqref{thiele} with an explicit finite difference method with step sizes $h=0.1$ for time and $l=1/12$ for space. Execution time around $0.2$ seconds.}
\label{fig:capsfloors}
\end{figure}

%%%%%%%%%%%%%%%%%%%%%%%%%%%%%%%%%%%%%%%%%%%%%%%%%%

\subsection{Binary endowment based on average interest rate}\label{bin.endow}
Let $E_1$ and $E_2$ be two endowments. The policy pays $E_1$ upon survival at expiry time $T$ if the average interest rate during the contract time is above $K$ and $E_2$ otherwise. Mathematically, the present value of this policy is given by
$$V(t,r_t,\overline{r}_t)= p_{\ast\ast}(t,T) \left[E_1\overline{U}_T^K(t,r_t,\overline{r}_t)+E_2 \left( \overline{U}_T^{-\infty}(t,r_t,\overline{r}_t)-\overline{U}_T^K(t,r_t,\overline{r}_t)\right)\right],$$
where here $\overline{U}_s^K$ is the function defined in \eqref{UbarK}.

In Figure \ref{fig:bin} we show the reserves for a binary endowment of $E_1 = \$ 150 \, 000$ if the average is above $K=4\%$ at the end of the contract and $E_2=\$ 100 \, 000$ if the average is below under the same model as in the previous examples. The reserve is indeed stochastic.

\begin{figure}[H]
\centering
  \includegraphics[scale=0.8]{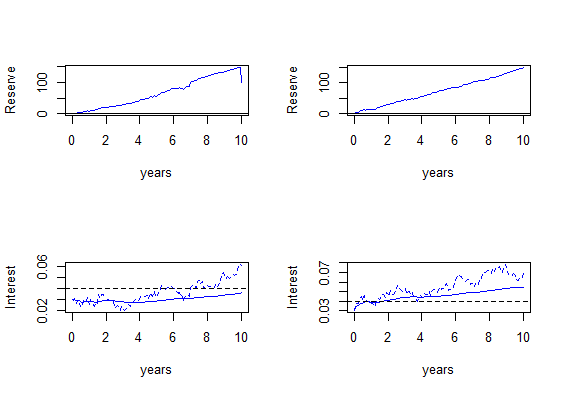}
   \caption{On top: Reserves for a binary endowment of $E_1=\$ 150 \,000$ in $\times 1 \, 000$ units if average rate is above $K=4\%$ and $E_2=\$ 100 \,000$ otherwise. On the bottom: the corresponding interest rate (random) outcomes (dashed line) and the running average during the contract (solid line). The premium obtained was $\pi= \$ 9\, 516.71$ when $r_0=3\%$. Vasicek parameters: $a=0.1$, $b=0.2$, $\sigma=0.01$.}
\label{fig:bin}
\end{figure}

%%%%%%%%%%%%%%%%%%%%%%%%%%%%%%%%%%%%%%%%%%%%%%%%%%
%%%%%%%%%%%%%%%%%%%%%%%%%%%%%%%%%%%%%%%%%%%%%%%%%%%%%%%%
%%%%%%%%%%%%%%%%%%%%%%%%%%%%%%%%%%%%%%%%%%%%%%%%%%%%%%%%

\subsection{A reinsurance treaty on pensions when insurer's average return is low}\label{pensionreins2}
Using the new model introduced in Section \ref{section:reserve} we can consider reinsurance agreements between the insurer and the reinsurer. If we assume that $r$ models the return on investments for the insurer, then their liabilities depend on $r$. If returns are low, then the pension liabilities increase, making it difficult to meet the requirements with their customers. One possibility to relax such risk could be to cede some of the risk (of low returns) to the reinsurer. In this example, we consider two polices; one for the insured and another one for the insurer (with the reinsurer). For a pension policy paying a yearly pension $P$ from $\hat{T}$ upon death (or high enough $T$) we consider
$$\dot{a}_\ast^{ins}(s) =g_\ast^{ins}(s,r_s,\overline{r}_s)= P, \quad s\in [\hat{T},T),$$
giving rise to
\begin{align*}
V^{ins}(t,r_t)\triangleq P\int_{\max\{t,\hat{T}\}}^T p_{\ast\ast}(t,s)U_s^{-\infty}(t,r_t)ds, \quad t\in [\hat{T},T]
\end{align*}
the value of the policy at each time. On the other hand, we can look at the performance of $\overline{r}_t$ from the contract start to the time pension payments start $\hat{T}$. If average returns are too low we can cede some of the liabilities to the reinsurer by purchasing a (stochastic) endowment (re)insurance in a similar fashion as in Section \ref{bin.endow}, which pays $100\rho\%$ of the value of the pension at time $t=\hat{T}$, that is the following (stochastic) endowment
$$\rho V^{ins}(\hat{T},r_{\hat{T}})\mathbb{I}_{\{\overline{r}_{\hat{T}}<K\hat{T}\}}.$$
The value of this endowment (re)insurance is thus given by
\begin{align*}
V^{re}(t,r_t,\overline{r}_t)=p_{\ast\ast}(t,\hat{T})\E_{\QQ}\left[e^{-\int_t^{\hat{T}} r_udu} \rho V^{ins}(\hat{T},r_{\hat{T}})\mathbb{I}_{\{\overline{r}_{\hat{T}}<K\hat{T}\}}\Big| \sigma(r_t)\otimes\sigma(\overline{r}_t)\right].
\end{align*}

The above conditional expectation is rather involved. This example shows why the PDE derived in Theorem \ref{ThielePDE} may be useful when looking at policies with stochastic payments where explicit expressions cannot directly be obtained. The reserve $V^{re}(t,x,y)$ is a function of three variables: time $t$, level of return $x$ and level of the average return $y$. It satisfies the following PDE
$$\partial_t V^{re} = xV^{re} +\mu(t)V^{re} - LV^{re},\quad t\in [0,\hat{T}],$$
where $L$ is the differential operator in \eqref{DiffOp}. The terminal condition is given by
$$V^{re}(\hat{T},x,y)=\rho V^{ins}(\hat{T},x)\mathbb{I}_{\{y<K\hat{T}\}}.$$
Here, $V^{ins}(t,x)$ is a function of two variables where again, $t$ is time, $x$ is the level of the return and it satisfies the following PDE
 $$\partial_t V^{ins} = xV^{ins}-P +\mu(t)V^{ins} - LV^{ins},\quad t\in [\hat{T},T],$$
with terminal condition $V^{ins}(T,x)=0$.

%%%%%%%%%%%%%%%%%%%%%%%%%%%%%%%%%%%%%%%%%%%%%%%%%%
%%%%%%%%%%%%%%%%%%%%%%%%%%%%%%%%%%%%%%%%%%%%%%%%%%
%%%%%%%%%%%%%%%%%%%%%%%%%%%%%%%%%%%%%%%%%%%%%%%%%%
%%%%%%%%%%%%%%%%%%%%%%%%%%%%%%%%%%%%%%%%%%%%%%%%%%

\end{document}